\documentclass[10pt,conference]{IEEEtran}
\IEEEoverridecommandlockouts
\usepackage[final]{graphicx}
\usepackage{cite}
\usepackage{amssymb}
\usepackage{amsmath}
\usepackage{amsthm}
\usepackage{amsfonts}
\usepackage{bm}
\usepackage{filecontents}

\title{Symbol-Level Precoding Design for Max-Min SINR in Multiuser MISO Broadcast Channels}

\author{Alireza~Haqiqatnejad,~Farbod~Kayhan,~and Bj\"{o}rn~Ottersten\\
Interdisciplinary Centre for Security, Reliability and Trust (SnT),~University of Luxembourg \\
email: \{alireza.haqiqatnejad,farbod.kayhan,bjorn.ottersten\}@uni.lu}

\newtheorem{theorem}{Theorem}
\newtheorem{lemma}{Lemma}
\newtheorem{property}{Property}

\newcommand{\Deee} {\mathrm{\pmb{\Delta}}}

\newcommand{\HHH} {\mathrm{\pmb{H}}}
\newcommand{\bbb}{\mathrm{\pmb{b}}}

\newcommand{\A}{\mathrm{\pmb{A}}}
\newcommand{\aaa}{\mathrm{\pmb{a}}}
\newcommand{\h}{\mathrm{\pmb{h}}}
\newcommand{\s}{\mathrm{\pmb{s}}}
\newcommand{\vvv}{\mathrm{\pmb{v}}}

\newcommand{\conv}{\mathrm{\pmb{conv}}}

\newcommand{\bd}{\mathrm{\pmb{bd}}}
\newcommand{\interior}{\mathrm{\pmb{int}}}
\newcommand{\ccc}{\mathrm{\pmb{c}}}
\newcommand{\x}{\mathrm{\pmb{x}}}
\newcommand{\yyy}{\mathrm{\pmb{y}}}
\newcommand{\uuu}{\mathrm{\pmb{u}}}
\newcommand{\C}{\mathbb{C}}

\begin{document}
\onecolumn
\maketitle

\IEEEpeerreviewmaketitle

\begin{abstract}
In this paper, we address the symbol level precoding (SLP) design problem under max-min SINR criterion in the downlink of multiuser multiple-input single-output (MISO) channels. First, we show that the distance preserving constructive interference regions (DPCIR) are always polyhedral angles (shifted pointed cones) for any given constellation point with unbounded decision region. Then we prove that any signal in a given unbounded DPCIR has a norm larger than the norm of the corresponding vertex if and only if the convex hull of the constellation contains the origin. Using these properties, we show that the power of the noiseless received signal lying on an unbounded DPCIR is an strictly increasing function of two parameters. This allows us to reformulate the originally non-convex SLP max-min SINR as a convex optimization problem. We discuss the loss due to our proposed convex reformulation and provide some simulation results.
\end{abstract}




\section{Introduction} \label{sec:intro}
Multiuser interference (MUI) is a major performance limiting factor in multiuser systems, which reduces the maximum reliable transmission rate of individual users. One approach to mitigate the MUI is to precompensate for its undesired effect on the received signal through some signal processing at the transmitter \cite{tb_sol_str}, which is known as multiuser precoding.

The multiuser precoding problem is usually expressed as a constrained optimization problem (see \cite{tb_sol_str},\cite{tb_coor} and the references therein). In general, the design problem aims at keeping a balance between some network-centric and user-centric objectives/requirements, depending on the network's operator strategy. Power and sum-rate are commonly regarded as network-centric criteria \cite{tb_coor}. 
On the other hand, as a user-centric criterion, signal-to-interference-plus-noise ratio (SINR) is an effective measure of quality-of-service (QoS) in multiuser interference channels \cite{tb_mul}. In particular, both bit error rate (BER) and capacity, which are two relevant criteria from a practical point of view, are closely related with maximizing SINR. Taking into account different types of optimization criteria, some well-known formulations for the multiuser precoding problem are power minimization with SINR constraints \cite{tb_sinr}, SINR balancing \cite{tb_conic}, and (weighted) sum-rate maximization \cite{tb_coor}. In this paper, we mainly focus on the SINR balancing problem through ensuring max-min fairness among users.

Conventional linear multiuser precoding techniques try to design the precoder in order to mitigate the MUI. This requires the knowledge of the instantaneous channel state information (CSI) to calculate the precoder matrix. Following the notion of constructive interference (CI), one can turn the MUI into a useful source of signal power instead of treating it as an unwanted distortion \cite{slp_chr}. Accordingly, in addition to CSI, the instantaneous data information (DI) of all users are used to design the precoder, which leads to introducing symbol-level precoding (SLP) \cite{slp_con}. When compared to conventional schemes, it has been shown that significant gains can be achieved, but at the cost of higher transmitter complexity \cite{slp_chr}. In SLP scheme, one may also form a virtual multicast formulation to directly find the optimal transmit vector, as proposed in \cite{slp_con}, instead of finding the precoder matrix.

The SINR balancing problem in multiuser multiple-input single-output (MISO) channels has been extensively investigated for conventional precoding techniques, and addressed in both multicast (single data stream) and broadcast (independent data streams) downlink scenarios. The problem is not convex in general and hence, alternate optimization approaches have been proposed. The authors in \cite{tb_mul} prove that the max-min SINR problem for downlink multicasting is NP-hard; however, it can be solved approximately through semidefinite relaxation. For downlink broadcast channels, it is shown in \cite{tb_conic} that the power minimization and the max-min SINR are inverse problems. Using this property, the SINR optimization problem can be solved through iteratively solving the power optimization. Furthermore, direct solutions for max-min SINR are provided in \cite{tb_conic} via conic optimization, where the problem is formulated as a quasiconvex standard generalized eigenvalue program (GEVP). Using the concept of uplink-downlink duality, in \cite{tb_sinr} it is verified that the global optimum of max-min SINR is equivalently obtained from solving a dual uplink problem, which has an easier-to-handle analytical structure.

Concerning SLP design in multiuser downlink channels, the optimization constraints push each user's (noiseless) received signal to a predefined region, called constructive interference region (CIR), enhancing (or guaranteeing a certain level of) the detection accuracy. This causes the constraints to depend on both the constellation set and the decision regions. In \cite{slp_con}, the non-convex SLP max-min SINR problem is solved using its relation to the power minimization via a bisection search. The proposed method is restricted to PSK constellations 
and has high computational complexity. This problem is also studied in \cite{slp_chr} and an alternate convex formulation is provided for PSK constellations. 
However, there is no general solution method or convex formulation for the SLP max-min SINR problem being valid for generic constellations of any order and shape.

In this paper, our goal is to find alternate convex formulations for the originally non-convex SLP max-min SINR problem, based on the definition of distance preserving constructive interference regions (DPCIR) \cite{slp_gen}. To obtain such reformulation, we first show that any DPCIR associated with a boundary constellation point is always a polyhedral angle, and hence unbounded. Then, we prove that any signal in a given unbounded DPCIR has a norm larger than the corresponding vertex under the necessary and sufficient condition that the constellation contains the origin in its convex hull. Based on these two results, we derive two alternate convex formulations for the SLP max-min SINR. This is done by noticing that the noise-free received signal at each user's receiver is an increasing function of two parameters. 

The remainder of this paper is organized as follows. In Section \ref{sec:sysmodel}, we describe our system model. In Section \ref{sec:cir}, we overview the DPCIRs and their properties. We discuss the SLP max-min SINR in Section \ref{sec:slp} and reformulate it as alternate convex problems. In Section \ref{sec:sim}, we provide some simulation results. Finally, we conclude the paper in Section \ref{sec:conc}.

{\bf{Notations:}} 
We use uppercase and lowercase bold-faced letters to denote matrices and vectors, respectively, and lowercase normal letters to denote scalars. For matrices and vectors, $[\cdot]^T$ denotes the transpose operator. For vectors, $\|\cdot\|$ represents the $l_2$ norm, and $\succeq$ (or $\succ$) denotes componentwise inequality. $\Re\{\cdot\}$ and $\Im\{\cdot\}$ denote the respectively real-part and imaginary-part operators. For any set $\mathcal{A}$, $|\mathcal{A}|$ denotes the cardinality of $\mathcal{A}$.


\section{System Model} \label{sec:sysmodel}

We consider the downlink of a multiuser MISO broadcast channel, where a base station (BS) transmits independent data streams to $K$ users. The BS is equipped with $N$ transmit antennas while each user has a single receive antenna. A complex channel vector is assumed between the BS's transmit antennas and the $k$-th user, which is denoted by $\h_k\in\C^{1\times N}$. It is further assumed that perfect channel knowledge is available to the BS.

At a given symbol time, $K$ independent symbols are to be sent to $K$ users (throughout the paper, we drop the symbol's time index to simplify the notation). We collect these symbols in users' symbol vector $\s=[s_1,\ldots,s_K]^T\in\C^{K\times1}$ with $s_k$ denoting the symbol intended for the $k$-th user. Each symbol $s_k$ is drawn from a finite equiprobable two-dimensional constellation set. Without loss of generality, we assume an $M$-ary constellation set $\chi=\{x_i|x_i\in\mathbb{C}\}_{i=1}^M$ with unit average power for all $K$ users. 
The user's symbol vector $\s$ is mapped onto $N$ transmit antennas. This is done by a symbol-level precoder yielding the BS's transmit vector $\uuu\in\C^{N\times1}$. The received signal at the $k$-th user's receiver is then
$r_k = \h_k\uuu+w_k$,
where $w_k\sim\mathcal{CN}(0,\sigma_k^2)$ is the complex additive white Gaussian noise at the $k$-th receiver. Again without loss of generality, we assume identical noise distributions across the receivers, i.e., $\sigma_k=\sigma, k=1,...,K$. From the received scalar $r_k$, the user $k$ may apply the maximum-likelihood (ML) decision rule to detect its own symbol $s_k$.

\section{Distance Preserving Constructive Interference Regions}\label{sec:cir}

In this section, we provide an overview of DPCIRs and their properties which will be useful in formulating the SLP design problem. Hereafter, we denote each complex-valued constellation point by its equivalent real-valued vector notation, hence the set of points in $\chi$ is denoted by $\{\x_i|\x_i\in\mathbb{R}^2\}_{i=1}^M$.

The DPCIRs can be described based on the hyperplane representation of ML decision regions \cite{slp_gen}. For the equiprobable constellation set $\chi$, the ML decision rule corresponds to the Voronoi regions of $\chi$ which are bounded by hyperplanes. For a given constellation point $\x_i$ and one of its neighboring points $\x_j$, the hyperplane separating the Voronoi regions of $\x_i$ and $\x_j$ is given by $\{\x\mid \x\in\mathbb{R}^2, \aaa_{i,j}^T \x=b_{i,j}\}$, where $\aaa_{i,j}=\x_i-\x_j$ (or any non-zero scalar multiplication of $\x_i-\x_j$), and $b_{i,j}=\aaa_{i,j}^T(\x_i+\x_j)/2$. This hyperplane indicates a decision boundary (Voronoi edge) between $\x_i$ and $\x_j$, which splits $\mathbb{R}^2$ plane into two halfspaces. The closed halfspace $\mathcal{H}_{i,j}=\{\x\mid \x\in\mathbb{R}^2, \aaa_{i,j}^T \x\geq b_{i,j}\}$ contains the decision region of $\x_i$, where $\aaa_{i,j}$ is the inward normal and $b_{i,j}$ determines the offset from the origin. The Voronoi region of $\x_i$ is then given by intersecting all such halfspaces, i.e.,
\setlength{\abovedisplayskip}{5pt}
\begin{equation}\label{eq:cap}
\begin{aligned}
\mathcal{D}_{i,\text{ML}}&=\bigcap_{\x_j\in\mathcal{S}_i}\mathcal{H}_{i,j}\\
&=\left\{\x\mid \x\in\mathbb{R}^2, \aaa_{i,j}^T \x\geq b_{i,j}, \forall \x_j\in\mathcal{S}_i\right\},
\end{aligned}
\end{equation}
where $\mathcal{S}_i$ denotes the set of neighboring points of $\x_i$ with $|\mathcal{S}_i|=M_i$ . Each Voronoi region can be either an unbounded or bounded polyhedron, depending on the relative location of $\x_i$ in $\chi$. It can be easily verified that all types of polyhedra (i.e., bounded or unbounded), and hence the Voronoi regions, are convex sets \cite{convex_boyd}. The Voronoi region \eqref{eq:cap} can be expressed in a more compact form as
\begin{equation}\label{eq:capc}
\mathcal{D}_{i,\mathrm{ML}}=\left\{\x\mid \x\in\mathbb{R}^2, \A_i \x\succeq \bbb_i\right\},
\end{equation}
where $\A_i\in\mathbb{R}^{M_i\times2}$ and $\bbb_i\in\mathbb{R}^{M_i}$ contain $\aaa_{i,j}^T$ and $b_{i,j}$, respectively, for all $\x_j\in\mathcal{S}_i$.

For any hyperplane $\{\x\mid \x\in\mathbb{R}^2, \aaa_{i,j}^T \x=b_{i,j}\}$, the set of points $\left\{\x\mid \x\in\mathbb{R}^2, \aaa_{i,j}^T \x=b_{i,j}+c_{i,j}, c_{i,j}\in\mathbb{R}_+\right\}$ represents a parallel hyperplane with the orthogonal distance $c_{i,j}/\|\aaa_{i,j}\|$ in the direction of $\aaa_{i,j}$. Let $d_{i,j}$ denote the distance between $\x_i$ and $\x_j$. Since the DPCIRs are defined so as to not decrease the original distances between the constellation points, the distance preserving margin is equal to $d_{i,j}/2$. Therefore, the DPCIR associated with $\x_i$ can be described as
\begin{equation}\label{eq:dpcirlinineq}
\mathcal{D}_{i,\mathrm{DP}}=\left\{\x\mid\x\in\mathbb{R}^2, \A_i \x\succeq \bbb_i+\ccc_{i,\text{DP}}\right\},
\end{equation}
where $\ccc_{i,\text{DP}}\in\mathbb{R}^{M_i}_+$ is the vector containing $d_{i,j}\|\aaa_{i,j}\|/2$ for all $\x_j\!\in\!\mathcal{S}_i$. Similar to $\mathcal{D}_{i,\mathrm{ML}}$, $\mathcal{D}_{i,\mathrm{DP}}$ is the intersection of a number of closed halfspaces and thus is a polyhedron. Furthermore, the bounding hyperplanes of $\mathcal{D}_{i,\mathrm{DP}}$ are parallel to their corresponding Voronoi edges.
It is straightforward to show that the following properties hold for DPCIRs:
\begin{property}\label{pro:0}
For any $\x_i\in\chi$ and $\x\in\mathcal{D}_{i,\mathrm{DP}}$, we have
\begin{itemize}
\item[i.]$\mathcal{D}_{i,\mathrm{DP}}\subseteq\mathcal{D}_{i,\mathrm{ML}}$.
\item[ii.]$\|\x-\yyy\|\geq\|\x_i-\x_j\|=d_{i,j}, \forall \x_j\in\chi, \yyy\in\mathcal{D}_{j,\mathrm{DP}}$.
\item[iii.]$\|\x-\x_j\|\geq\|\x_i-\x_j\|, \forall \x_j\in\chi$, where equality holds only when $\x=\x_i$.
\end{itemize}
\end{property}
From the constellation set $\chi$, one can easily derive its convex hull $\conv\chi$, i.e., the smallest convex set containing $\chi$ (see Fig. \ref{fig:1}). The set of points belonging to the boundary of $\conv\chi$ is denoted by $\bd\chi$, and the set of interior points of $\conv\chi$, i.e., $\conv\chi\backslash\bd\chi$, is shown by $\interior\chi$. It follows from \eqref{eq:dpcirlinineq} that if $\mathcal{D}_{i,\mathrm{ML}}$ is bounded, then $\mathcal{D}_{i,\mathrm{DP}}=\x_i$. On the other hand, for an unbounded $\mathcal{D}_{i,\mathrm{ML}}$, the associated $\mathcal{D}_{i,\mathrm{DP}}$ is an unbounded polyhedron (more specifically, a polyhedral angle as depicted in Fig. \ref{fig:1}) which is uniquely characterized using the two following lemmas.
\begin{figure}
\centering
\includegraphics[width=.45\columnwidth]{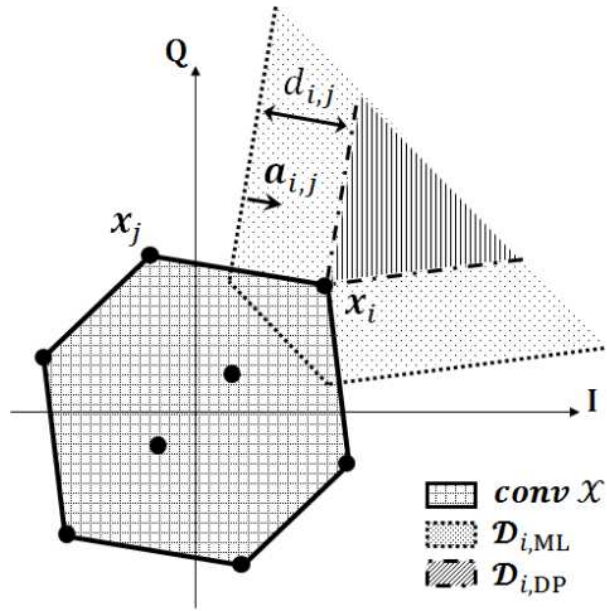}
\caption{The AWGN-optimized 8-ary constellation and one of its boundary points $\x_i$ with unbounded Voronoi region. The associated DPCIR is a polyhedral angle with two infinite edges starting from $\x_i$.}
\label{fig:1}
\end{figure}
\begin{lemma}\label{lem:1}
A point $\x_i\in\chi$ lies on the boundary of (or is a vertex of) $\conv\chi$ iff its Voronoi region $\mathcal{D}_{i,\mathrm{ML}}$ is unbounded \cite[Lemma 2.2]{vor_diag}.
\end{lemma}
\begin{lemma}\label{lem:2}
For every $\x_i\in\chi$ with unbounded $\mathcal{D}_{i,\mathrm{ML}}$, $\mathcal{D}_{i,\mathrm{DP}}$ is a polyhedral angle with a vertex at $\x_i$, and each of its edges is perpendicular to one of the two line segments connecting $\x_i$ to its two neighboring points on $\bd\chi$.
\end{lemma}
\begin{proof}
See Appendix \ref{app:lem2}.
\end{proof}
Lemma \ref{lem:2} implicitly states that neither changing the location of any constellation point $\x_j\in\interior\chi$ nor adding a new constellation point on $\bd\chi$ does not affect $\mathcal{D}_{i,\mathrm{DP}}$ for any $\x_i\in\bd\chi$, as they both keep the direction of $\aaa_{i,j}$ unchanged for all $\x_j\in\mathcal{S}_i\cap\bd\chi$. This leads to the following two lemmas.
\begin{lemma}\label{lem:3}
For any constellation point $\x_i\in\chi$, we have $\|\x\|\geq\|\x_i\|,\forall\x\in\mathcal{D}_{i,\mathrm{DP}}$ iff $\conv\chi$ contains the origin. Equality is achieved only when $\x=\x_i$.
\end{lemma}
\begin{proof}
See Appendix \ref{app:lem3}.
\end{proof}
\begin{lemma}\label{lem:4}
If $\mathbf{0}\notin\conv\chi$, there exists at least one constellation point $\x_l\in\chi$ for which for any $\x\in\mathcal{D}_{l,\mathrm{DP}}$, $\mathbf{0}\notin\conv\tilde{\chi}_{\x_l,\x}$, where $\tilde{\chi}_{\x_l,\x}=\chi\cup\{\x\}$.
\end{lemma}
\begin{proof}
If $\mathbf{0}\notin\conv\chi$, for any $\x_i\in\chi$ and any $\x\in\mathcal{D}_{i,\mathrm{DP}}$ with $\tilde{\chi}_{\x_i,\x}=\chi\cup\{\x\}$, let define $\mathcal{C}_i=\underset{\x\in\mathcal{D}_{i,\mathrm{DP}}}\bigcup\conv\tilde{\chi}_{\x_i,\x}$. Having $\conv\chi\subseteq\conv\tilde{\chi}_{\x_i,\x}$, it follows from the definition of convex hull that $\conv\chi=\underset{\x_i\in\chi}\bigcap\mathcal{C}_i$. If $\mathbf{0}\in\mathcal{C}_i, \forall\x_i\in\chi$, then $\mathbf{0}\in\conv\chi$ which contradicts our assumption. Hence there must exist at least one constellation point, say $\x_l$, for which $\mathcal{C}_l$ and therefore none of $\conv\tilde{\chi}_{\x_l,\x}, \forall\x\in\mathcal{D}_{l,\mathrm{DP}}$ contains the origin, as required.
\end{proof}

To proceed, it is more convenient to express the linear inequalities of \eqref{eq:dpcirlinineq} by an equivalent set of linear equations as
\begin{equation}\label{eq:dpcirlineq}
\mathcal{D}_{i,\text{DP}}=\Big\{\x\!\mid\!\x\!\in\!\mathbb{R}^2\!, \A_i \x=\bbb_i+\ccc_{i,\text{DP}}+\Deee_i, \Deee_i\!\in\!\mathbb{R}^{M_i}_+ \Big\}.
\end{equation}
The linear equations in \eqref{eq:dpcirlineq} indicate that any $\x\in\mathcal{D}_{i,\mathrm{DP}}$ can be specified as the intersection point of $M_i$ hyperplanes, each of which is parallel to a boundary hyperplane of $\mathcal{D}_{i,\mathrm{DP}}$ but has a different offset due to the term $\Deee_i$. Finally, we state the following theorem which is the main result of this section.
\begin{theorem}\label{thm:1}
For any constellation point $\x_i\in\chi$ with $\mathcal{D}_{i,\mathrm{DP}}$ as expressed in \eqref{eq:dpcirlineq}, function $f(\x)=\|\x\|$ over its domain $\mathcal{D}_{i,\mathrm{DP}}$ is a strictly monotonic increasing function of the elements of $\Deee_i$ iff $\conv\chi$ contains the origin.
\end{theorem}
\begin{proof}
See Appendix \ref{app:thm}
\end{proof}
\section{Symbol-level Precoding Design Problem}\label{sec:slp}

A symbol-level precoder designs the vector to be transmitted at each symbol time via solving a constrained optimization problem. The solution of the SLP problem, i.e., the transmit vector $\uuu$, is in general a function of instantaneous DI and CSI as well as the set of given system constraints or user-specific requirements. When power is a strict system restriction on the downlink transmission, fairness might be a relevant design criterion \cite{tb_conic}. In particular, we are interested in the SLP max-min SINR problem subject to a total power constraint $P_\text{max}$ which aims at maximizing the worst SINR among all users. Assuming the CIRs to be distance preserving, the problem is not convex in its original form. In this section, we derive two alternate convex formulations for this problem. This is done by noticing that the noise-free received signal at each user is a increasing function of two parameters. 

For any user $k=1,...,K$, the symbol $s_k$ corresponds to one of the points $\{\x_i\}_{i=1}^M$ in $\chi$. In the following, we denote by $i_k$ the index of the constellation point corresponding to $s_k$, i.e., $[\Re\{s_k\},\Im\{s_k\}]^T=\x_{i_k}$, where $i_k\in\{1,...,M\}$. Furthermore, vectors $\uuu$ and $\h_k$ are rearranged as $\tilde{\uuu}=[\Re\{\uuu\},\Im\{\uuu\}]^T\in\mathbb{R}^{2N\times1}$ and $\HHH_k=\left[[\Re\{\h_k\},\Im\{\h_k\}]^T,[-\Im\{\h_k\},\Re\{\h_k\}]^T\right]\in\mathbb{R}^{2\times2N}, k=1,...,K$, respectively, such that $\HHH_k\tilde{\uuu}$ represents the noise-free received signal at the $k$-th user's receiver.

The symbol-level SINR for user $k$ is proportional to the instantaneous received power by the $k$-th receiver at each symbol time (recall that the same noise variance $\sigma^2$ is assumed for all $K$ users). Accordingly, the DPCIR-based SLP max-min SINR problem can be formulated as
\abovedisplayskip=1pt
\begin{equation}\label{eq:sb}
\begin{aligned}
\underset{\tilde{\uuu}, \{\Deee_{i_k}\}_{k=1}^K}{\text{maximize}} & \quad\underset{k}{\min}\left\{\tilde{\uuu}^T\HHH_k^T\HHH_k\tilde{\uuu}\right\}_{k\in\mathcal{K}} \\
\text{subject to} & \quad \A_{i_k}\HHH_k\tilde{\uuu}=\sigma\;(\bbb_{i_k}+\ccc_{i_k}+\Deee_{i_k}), k=1,...,K,\\
& \quad \Deee_{i_k}\succeq\mathbf{0}_k, k=1,...,K,\\
& \quad \tilde{\uuu}^T\tilde{\uuu}\leq P_\text{max},
\end{aligned}
\end{equation}
\fontdimen2\font=3.225pt
where the index set $\mathcal{K}\!=\!\{k|k=1,...,K, \x_{i_k}\!\in\!\bd\chi\}$ refers to those users with a symbol in the boundary of their corresponding constellation, and $\mathbf{0}_k$ denotes an all-zeros vector of appropriate dimension. By introducing a slack variable $t$, one can recast \eqref{eq:sb} as
\fontdimen2\font=3.33pt
\abovedisplayskip=1pt
\begin{equation}\label{eq:sbt}
\begin{aligned}
\underset{\tilde{\uuu}, \{\Deee_{i_k}\}_{k=1}^K}{\text{maximize}} & \quad t \\
\text{subject to} & \quad \A_{i_k}\HHH_k\tilde{\uuu}=\sigma\;(\bbb_{i_k}+\ccc_{i_k}+\Deee_{i_k}), k=1,...,K,\\
& \quad \Deee_{i_k}\succeq\mathbf{0}_k, k=1,...,K,\\
& \quad \tilde{\uuu}^T\HHH_k^T\HHH_k\tilde{\uuu}\geq t, k\in\mathcal{K},\\
& \quad \tilde{\uuu}^T\tilde{\uuu}\leq P_\text{max},
\end{aligned}
\end{equation}
which is not convex due to the third set of constraints. As a consequence of Lemma \ref{lem:1} and Lemma \ref{lem:2}, and with respect to \eqref{eq:dpcirlineq}, any point in $\mathcal{D}_{{i_k},\text{DP}}$ can be uniquely specified by $\Deee_{i_k}=[\delta_{i_k,1},\delta_{i_k,2}]^T\in\mathbb{R}^2_+$ for all $\x_{i_k}\in\bd\chi$. It follows from Theorem \ref{thm:1} that $\tilde{\uuu}^T\HHH_k^T\HHH_k\tilde{\uuu}=\|\HHH_k\tilde{\uuu}\|^2$ is strictly increasing in each element of $\Deee_{i_k}, \forall k\in\mathcal{K}$, i.e., assuming either $\delta_{i_k,1}$ or $\delta_{i_k,2}$ to be fixed, $\tilde{\uuu}^T\HHH_k^T\HHH_k\tilde{\uuu}$ is a monotonically increasing function of the other. This suggests that if the optimal value of one of the elements, e.g., $\delta_{i_k,1}$, is given for all $k\in\mathcal{K}$, then the optimization \eqref{eq:sbt} is equivalent to the convex problem
\abovedisplayskip=1pt
\begin{equation}\label{eq:sbtfixed}
\begin{aligned}
\underset{\tilde{\uuu}, \delta_{i_k,2}:k\in\mathcal{K}}{\text{maximize}} & \quad t \\
\text{subject to} & \quad \A_{i_k}\HHH_k\tilde{\uuu}=\sigma\;(\bbb_{i_k}+\ccc_{i_k}+\Deee_{i_k}), k=1,...,K,\\
& \quad \delta_{i_k,2}\geq t, k\in\mathcal{K}, \;\Deee_{i_k}=\mathbf{0}_k, k\notin\mathcal{K},\\
& \quad \tilde{\uuu}^T\tilde{\uuu}\leq P_\text{max},
\end{aligned}
\end{equation}
where the parameter $\delta_{i_k,2}$ is substituted for $ \tilde{\uuu}^T\HHH_k^T\HHH_k\tilde{\uuu}$ in \eqref{eq:sbt}.
In fact, achieving the optimum of \eqref{eq:sbt} requires an exhaustive search over all possible (non-negative) values of $\delta_{i_k,1},\forall k\in\mathcal{K}$ and solving \eqref{eq:sbtfixed} for each choice of $\delta_{i_k,1}$. The optimal solution is then obtained by picking $\delta_{i_k,1}$ for which the objective function is maximum among all other choices. Practically speaking, due to the power limitation induced by $P_\text{max}$, one can bound and discretize the search interval to choose $\delta_{i_k,1},\forall k\in\mathcal{K}$ from a finite set, which of course leads to a sub-optimal solution. The gap to the optimal solution depends on whether the search interval includes the optimal value, and also on the step size of discretization. In general, the smaller the step size is, the higher the computational complexity will be.

Another alternate, but not equivalent, convex formulation for problem \eqref{eq:sbt} is to jointly optimize $\delta_{i_k,1}$ and $\delta_{i_k,2}$ for all $k\in\mathcal{K}$, i.e.,\fontdimen2\font=3.33pt
\abovedisplayskip=1pt
\begin{equation}\label{eq:sbdu}
\begin{aligned}
\underset{\tilde{\uuu}, \{\Deee_{i_k}\}_{k=1}^K}{\text{maximize}} & \quad t \\
\text{subject to} & \quad \A_{i_k}\HHH_k\tilde{\uuu}=\sigma\;(\bbb_{i_k}+\ccc_{i_k}+\Deee_{i_k}), k = 1,...,K,\\
& \quad \Deee_{i_k}\succeq t\:\mathbf{1}_k, k\in\mathcal{K}, \;\Deee_{i_k}= \mathbf{0}_k, k\notin\mathcal{K},\\
& \quad \tilde{\uuu}^T\tilde{\uuu}\leq P_\text{max}.
\end{aligned}
\end{equation}
where $\mathbf{1}_k$ is an all-ones vector. The optimal solution of this problem can be regarded as a lower bound on the optimum of SLP max-min SINR. It should be noted that for PSK constellations, problem \eqref{eq:sbdu} is equivalent to the alternate convex formulation for the SLP SINR balancing provided in \cite{slp_chr}.

\section{Simulation Results}\label{sec:sim}
In this section, we provide the simulation results to evaluate the performance of the two proposed alternate convex formulations for the SLP max-min SINR problem. In the simulations, we have considered a downlink multiuser scenario with $N=K=4$ and $\sigma=1$, where the BS employs the AWGN-optimized 8-ary constellation for all users. For any user $k$, the complex channel vector $\h_k$ follows an i.i.d. complex Gaussian distribution with zero-mean and unit variance. The results are averaged over $500$ symbol slots.

Figure \ref{fig:2} shows the optimized minimum SINR across the users obtained from the joint optimization problem \eqref {eq:sbdu} and from an exhaustive search over interval [0,2.5] with step size 0.5 for all possible combinations of $\delta_{i_k,1},\forall k\in\mathcal{K}$. For the exhaustive search, the number of convex problems to be solved in every symbol time is of order $5^K$. As it can be observed, the loss due to the joint, but convex, optimization is not significant (around 1-1.5 dBW). This loss is, however, a consequence of highly reducing the computational complexity to solving only one convex problem in each symbol time.
\begin{figure}
\centering
\includegraphics[width=.5\columnwidth]{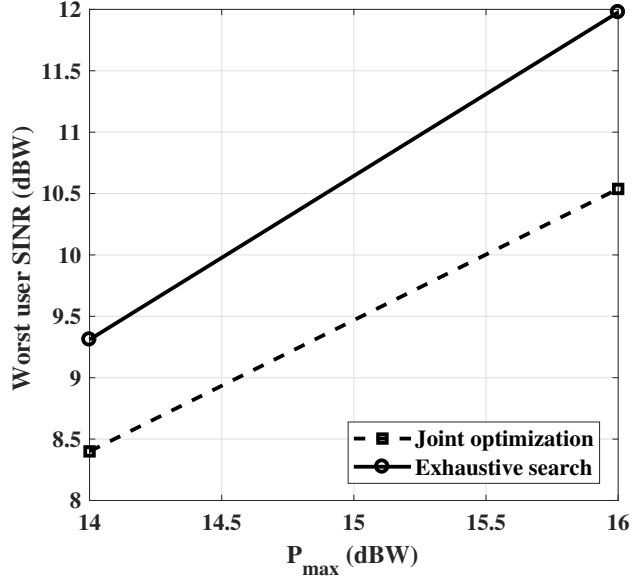}
\caption{The worst user SINR versus total power constraint.}
\label{fig:2}
\end{figure}
\section{Conclusion}\label{sec:conc}
In this paper, we addressed the problem of SLP SINR balancing with max-min fairness criterion in a downlink multiuser MISO channel. 
The original formulation of this problem is known to be non-convex. For the so-called DPCIRs, we proposed two alternate convex formulations. Both formulations are based on the observation that the noise-free received signal power is an strictly increasing function of two parameters for any given unbounded DPCIR, under the necessary and sufficient condition that the convex hull of the constellation contains the origin. The first formulation is solvable via an exhaustive search and even though it provides the optimal solution, it is computationally expensive to be implemented in a realistic scenario. The second formulation, though sub-optimal, reduces the problem to a convex optimization.


\section*{Acknowledgment}
The authors are supported by the Luxembourg National Research Fund under CORE Junior project: C16/IS/11332341 Enhanced Signal Space opTImization for satellite comMunication Systems (ESSTIMS).

\appendices
\section{}\label{app:1}
\subsection{Proof of Lemma \ref{lem:2}}\label{app:lem2}
The intersection of finitely many closed halfspaces is an unbounded polyhedron if and only if the outward normals to the associated boundary hyperplanes lie on a single closed halfspace \cite[p. 20, Theorem 4]{convex_poly}. Accordingly, for any $\x_i\in\chi$ with unbounded $\mathcal{D}_{i,\mathrm{ML}}$, the outward normal vectors $-\aaa_{i,j}, \forall\x_j\in\mathcal{S}_i$ lie on a single halfspace. Since $\mathcal{D}_{i,\mathrm{DP}}$ has the same set of outward normal vectors $-\aaa_{i,j}, \forall\x_j\in\mathcal{S}_i$, it is also unbounded. An unbounded polyhedron is uniquely determined from its vertices and the directions of its infinite edges \cite[p. 31, Theorem 4]{convex_poly}. It is straightforward to check that $\x_i$ is the unique solution of $\A_i \x=\bbb_i+\ccc_{i,\text{DP}}$, i.e., all the hyperplanes have a common intersection point $\x_i$. Therefore, $\mathcal{D}_{i,\mathrm{DP}}$, which is given by the solution set of $\A_i \x\succeq\bbb_i+\ccc_{i,\text{DP}}$, has a single vertex at $\x_i$ and two infinite edges, i.e., a polyhedral angle. In addition, since any two neighboring points share a common Voronoi edge, the two infinite edges of $\mathcal{D}_{i,\mathrm{DP}}$ correspond to the two neighboring points of $\x_i$ on $\bd\chi$ (i.e., $\mathcal{S}_i\cap\bd\chi$) with unbounded Voronoi regions. Each infinite edge of $\mathcal{D}_{i,\mathrm{DP}}$ is then parallel to a hyperplane with normal vector $\aaa_{i,j}=\x_i-\x_j$, where $\x_j\in\mathcal{S}_i\cap\bd\chi$; therefore it is perpendicular to $\x_i-\x_j$. This completes the proof.
\subsection{Proof of Lemma \ref{lem:3}}\label{app:lem3}
\emph{Sufficiency}: Having $\mathbf{0}\in\conv\chi$, let further assume that $\mathbf{0}\in\chi$. This assumption, as mentioned earlier, does not have any impact on $\mathcal{D}_{i,\mathrm{DP}}$ for any $\x_i\in\bd\chi$, regardless of whether $\mathbf{0}\in\bd\chi$ or $\mathbf{0}\in\interior\chi$. By substituting $\x_j=\mathbf{0}$ in Property \ref{pro:0} (iii), for all $\x_i\in\chi$ we have $\|\x\|\geq\|\x_i\|,\forall\x\in\mathcal{D}_{i,\mathrm{DP}}$. This completes the proof of sufficiency.

We use the following well-known property of convex sets to prove the necessity.
\begin{property}\label{pro:1}
$\vvv_o$ is the minimum distance vector from the origin to the convex set $\mathcal{V}$ iff for any vector $\vvv\in\mathcal{V}$ we have $\vvv_o^T\vvv\geq\vvv_o^T\vvv_o$, with equality for $\vvv$ lying on the hyperplane orthogonal to $\vvv_o$ \cite[p. 69, Theorem 1]{opt_vec}.
\end{property}
\emph{Necessity:} By contradiction, if $\mathbf{0}\notin\conv\chi$, let assume a new constellation set $\tilde{\chi}$ having all the points of $\chi$ including the origin, i.e., $\tilde{\chi}=\chi\cup\{\mathbf{0}\}$, hence $\conv\chi\subset\conv\tilde{\chi}$. Clearly, $\mathbf{0}\in\bd\tilde{\chi}$ and according to Lemma \ref{lem:2}, there always exist exactly two constellation points on $\bd\tilde{\chi}$ that $\mathbf{0}$ contributes to their DPCIRs. Suppose $\x_l$ be one of these points with $\mathcal{D}_{l,\mathrm{DP}}$ and $\tilde{\mathcal{D}}_{l,\mathrm{DP}}$ denoting its associated DPCIR relative to $\chi$ and $\tilde{\chi}$, repectively. We denote by $\tilde{\mathcal{S}}_l$ the set of neighboring points of $\x_l$ in $\tilde{\chi}$. Let $\mathcal{H}_{l,o}=\left\{\x\mid\x\in\mathbb{R}^2, \x_l^T\x\geq \x_l^T\x_l\right\}$ be the distance preserving halfspace from $\mathbf{0}$ to $\x_l$. Since $\mathbf{0}\in\tilde{\mathcal{S}}_l$, we have $\tilde{\mathcal{D}}_{l,\mathrm{DP}}=\mathcal{H}_{l,o}\cap\mathcal{D}_{l,\mathrm{DP}}\neq\mathcal{D}_{l,\mathrm{DP}}$, i.e., the halfspace $\mathcal{H}_{l,o}$ does not contain $\mathcal{D}_{l,\mathrm{DP}}$. Hence, $\left\{\x\mid\x\in\mathbb{R}^2, \x_l^T\x=\x_l^T\x_l\right\}$ is not a supporting hyperplane for $\mathcal{D}_{l,\mathrm{DP}}$ at $\x_l$ \cite[p. 51]{convex_boyd}. This implies that there exist some $\x\in\mathcal{D}_{l,\mathrm{DP}}$ for which $\x_l^T\x<\x_l^T\x_l$. According to Property \ref{pro:1} (which gives a necessary and sufficient condition), $\x_l$ is not the minimum distance vector from the origin in $\mathcal{D}_{l,\mathrm{DP}}$. Consequently, $\|\x\|\geq\|\x_l\|$ does not hold for some $\x\in\mathcal{D}_{l,\mathrm{DP}}$ which contradicts $\|\x\|\geq\|\x_l\|,\forall\x\in\mathcal{D}_{l,\mathrm{DP}}$.
\subsection{Proof of Theorem \ref{thm:1}}\label{app:thm}
\emph{Sufficiency:}
Suppose $\mathbf{0}\in\conv\chi$. Assuming a constellation point $\x_i\in\chi$ and its DPCIR $\mathcal{D}_{i,\mathrm{DP}}$, let $\yyy_1$ and $\yyy_2$ be two points in $\mathcal{D}_{i,\mathrm{DP}}$ such that $\A_i \yyy_1=\bbb_i+\ccc_{i,\text{DP}}+\Deee_{i,1}$ and $\A_i \yyy_2=\bbb_i+\ccc_{i,\text{DP}}+\Deee_{i,2}$ with $\Deee_{i,1},\Deee_{i,2}\in\mathbb{R}^{M_i}_+$ and $\Deee_{i,1}\prec\Deee_{i,2}$. Let consider a new constellation $\tilde{\chi}=\chi\cup\{\yyy_1\}$. It is clear that $\conv\chi\subseteq\conv\tilde{\chi}$, and therefore $\mathbf{0}\in\conv\tilde{\chi}$. The DPCIR of $\yyy_1$ can be described as $\mathcal{D}_{\yyy_1,\text{DP}}=\big\{\x\mid\x\in\mathbb{R}^2, \A_i \x=\bbb_i+\ccc_{i,\text{DP}}+\Deee_{i,1}+\Deee_1, \Deee_1\in\mathbb{R}^{M_i}_+ \big\}$. Let $\bar{\Deee}={\Deee}_{i,2}-{\Deee}_{i,1}$, then $\A_i \yyy_2=\bbb_i+\ccc_{i,\text{DP}}+\Deee_{i,1}+\bar{\Deee}, \bar{\Deee}\in\mathbb{R}^{M_i}_{++}$, which means that $\yyy_2\in\mathcal{D}_{\yyy_1,\text{DP}}$. As a consequence, from Lemma \ref{lem:3}, we have $\|\yyy_1\|<\|\yyy_2\|$ and the proof of sufficiency is complete.

\emph{Necessity:} 
By contradiction, suppose $\mathbf{0}\notin\conv\chi$. Then, based on Lemma \ref{lem:4}, there exists a constellation point $\x_l$ for which $\mathbf{0}\notin\conv\tilde{\chi}_{\x_l,\x}, \forall\x\in\mathcal{D}_{l,\mathrm{DP}}$. Let $\yyy_1\in\mathcal{D}_{l,\mathrm{DP}}$, then $\A_l \yyy_1=\bbb_l+\ccc_{l,\text{DP}}+\Deee_{l,1}$ with $\Deee_{l,1}\in\mathbb{R}^{M_l}_+$. The DPCIR associated with $\yyy_1$ can be expressed as $\mathcal{D}_{\yyy_1,\text{DP}}=\big\{\x\mid\x\in\mathbb{R}^2, \A_l \x=\bbb_l+\ccc_{l,\text{DP}}+\Deee_{l,1}+\Deee_1, \Deee_1\in\mathbb{R}^{M_l}_+ \big\}$. Since $\mathbf{0}\notin\conv\tilde{\chi}_{\x_l,\yyy_1}$, it follows from Lemma \ref{lem:3} and Property \ref{pro:1} that there exists $\yyy_2\in\mathcal{D}_{\yyy_1,\text{DP}}$ such that $\A_l \yyy_2=\bbb_l+\ccc_{l,\text{DP}}+\Deee_{l,1}+\bar{\Deee}, \bar{\Deee}\in\mathbb{R}^{M_l}_{++}$, for which $\|\yyy_2\|<\|\yyy_1\|$. But ${\Deee}_{l,1}+\bar{\Deee}={\Deee}_{l,2}$ yields ${\Deee}_{l,2}\succ{\Deee}_{l,1}$ which is a contradiction. This completes the proof.

\end{document}